\documentclass{article}
\usepackage{amsmath,amssymb,amsfonts,amsthm,epsfig,color,graphicx,eucal,mathrsfs,tikz,float}
\newtheorem{thm}{Theorem}

\newtheorem{remark}{Remark}
\newtheorem{lemma}{Lemma}
\newtheorem{definition}{Definition}
\newtheorem{Observation}{Observation}
\newtheorem{prop}{Proposition}
\newtheorem{claim}{Claim}

\title{Algorithmic Aspects of Some Variants of Domination in Graphs}
\author{J. Pavan Kumar and P.Venkata Subba Reddy\\
	Department of Computer Science and Engineering \\ 
	National Institute of Technology \\ 
	Warangal, Telangana 506004, India. \\
	jp.nitw@gmail.com, pvsr@nitw.ac.in}

\date{}
\begin{document}
\maketitle

\begin{center}
\end{center}
\vspace{3mm}

\begin{abstract}
A set $S \subseteq V$ is a \textit{dominating set} in $G$ if for every $u \in V \setminus S$, there exists $v \in S$ such that $(u,v) \in E$, i.e., $N[S] = V$. A dominating set $S$ is an \emph{Isolate Dominating Set} (\textit{IDS}) if the induced subgraph $G[S]$ has at least one isolated vertex. It is known that Isolate Domination Decision  problem (IDOM) is NP-complete for bipartite graphs. In this paper, we extend this by showing that the IDOM is NP-complete for split graphs and perfect elimination bipartite graphs, a subclass of bipartite graphs. A set $S \subseteq V$ is an \textit{independent set} if $G[S]$ has no edge. A set $S \subseteq V$ is a \textit{secure dominating set} of $G$ if, for each vertex $u \in V \setminus S$, there exists a vertex $v \in S$ such that $ (u,v) \in E $ and $(S \setminus \{v\}) \cup \{u\}$ is a dominating set of $G$. In addition, we initiate the study of a new domination parameter called, \textit{independent secure domination}. A set $S\subseteq V$ is an \textit{Independent Secure Dominating Set} (\textit{InSDS}) if $S$ is an independent set and a secure dominating set of $G$. The minimum size of an InSDS in $G$ is called the \textit{independent secure domination number} of $G$ and is denoted by $\gamma_{is}(G)$. Given a graph $ G$ and a positive integer $ k,$ the InSDM problem is to check whether $ G $ has an independent secure dominating set of size at most $ k.$ We prove that InSDM is NP-complete for bipartite graphs and linear time solvable for bounded tree-width graphs and threshold graphs, a subclass of split graphs. The MInSDS problem is to find an independent secure dominating set of minimum size, in the input graph. Finally, we show that the MInSDS problem is APX-hard for graphs with maximum degree $5.$
\end{abstract}

\vskip.5cm

\section{Introduction}
In this paper, every graph $G=(V,E)$ considered is finite, simple (i.e., without self-loops and multiple edges) and undirected with vertex set $V$ and edge set $E$.  
For a vertex $v \in V$, the (\textit{open}) \textit{neighborhood} of $v$ in $G$ is $N(v)= \{u \in V:(u,v) \in E$\}, the \textit{closed neighborhood} of $v$ is defined as $N[v]=N(v) \cup \{v\}$. If $S \subseteq V$, then the (open) neighborhood of $S$ is the set $N(S) = \cup_{v \in S} N(v)$. The closed neighborhood of $S$ is $N[S] = S \cup N(S)$. The \textit{degree} of a vertex $v$ is the size of the set $N(v)$ and is denoted by $d(v)$. If $d(v) = 0$, then $v$ is called an \textit{isolated vertex} of $G$. If $d(v) = 1$, then $v$ is called a \textit{pendant vertex}. For a graph $G=(V,E),$ and a set $S \subseteq V,$ the subgraph of $G$ \textit{induced} by $S$ is defined as $G[S]=(S,E_S)$, where $E_S=\{(x,y)$ $:$ $x,y \in S$ and $(x,y)\in E \}$. A \textit{spanning subgraph} is a subgraph that contains all the vertices of the graph. If $G[S]$ is a complete subgraph of $G$, then it is called a \textit{clique} of $G$. A set $S \subseteq V$ is an \textit{independent set} if $G[S]$ has no edges. Let $ S \subseteq V(G)$ and $v$ be a vertex in $S$, then the $S$-\textit{external private neighborhood} of $v$ denoted \textit{epn}$(v,S)$ is defined as $\{w:  w \in V(G) \setminus S$ and $N(w) \cap S = \{ v \} \}$.\par
A \textit{split graph} is a graph in which the vertices can be partitioned into a clique and an independent set. For a bipartite graph $G=(X,Y,E)$, an edge $(u, v) \in E$ is \textit{bisimplicial} if $N(u) \cup N(v)$ induces a complete bipartite subgraph in $G$. Let $(e_1,e_2,\ldots,e_k)$ be an ordering of pairwise non-adjacent edges of $G$ (not necessarily all edges of $G$). Let $S_i$ be the set of endpoints of edges $e_1,e_2,\ldots, e_i$ and $S_0=\phi$. An ordering $(e_1,e_2,\ldots,e_k)$ is a \textit{perfect edge elimination ordering} for $G$, if $G[(X \cup Y) \setminus S_k]$ has no edge and each edge $e_i$ is bisimplicial in the remaining induced subgraph $G[(X \cup Y) \setminus S_{i-1}]$. A graph $G$ is a \textit{perfect elimination bipartite graph} if $G$ admits a perfect edge elimination ordering and this subclass of bipartite graphs has been introduced by Golumbic and Goss in \cite{gol}. \par 
In a graph $ G=(V,E) $, a set $S \subseteq V$ is a \textit{Dominating Set} (\textit{DS}) in $G$ if for every $u \in V \setminus S$, there exists $v \in S$ such that $(u,v) \in E$, i.e., $N[S] = V$. The minimum size of a dominating set in $G$ is called the \textit{domination number} of $G$ and is denoted by $\gamma(G)$. Given a graph $ G=(V,E)$ and a positive integer $ k,$ the DOMINATION DECISION problem is to check whether $ G $ has a dominating set of size at most $ k.$ The DOMINATION DECISION problem is known to be NP-complete \cite{garey}. A set $S \subseteq V$ is an \textit{Independent Dominating Set} (\textit{InDS}) of $G$ if $S$ is an independent set and every vertex not in $S$ is adjacent to a vertex in $S$. The minimum size of an InDS in $G$ is called the \textit{independent domination number} of $G$ and is denoted by $i(G)$. The literature on various domination parameters in graphs has been surveyed in \cite{Haynes1,Haynes2}.\par 
An important domination parameter called secure domination has been introduced by E.J. Cockayne in \cite{pog}. A dominating set $S \subseteq V$ is a \textit{Secure Dominating Set} (\textit{SDS}) of $G$, if for each vertex $u \in V \setminus S$, there exists a neighboring vertex $v$ of $u$ in $S$ such that $(X \setminus \{v\}) \cup \{u\}$ is a dominating set of $G$ (in which case $v$ is said to \textit{defend} $u$). The minimum size of a SDS in $G$ is called the \textit{secure domination number} of $G$ and is denoted by $\gamma_s(G)$. Given a graph $ G=(V,E)$ and a positive integer $ k,$ the SDOM problem is to check whether $ G $ has a secure dominating set of size at most $ k.$ The SDOM problem is known to be NP-complete for general graphs \cite{dev} and remains NP-complete even for various restricted families of graphs such as bipartite graphs and split graphs \cite{osd}. Another domination parameter called isolate domination has been introduced by Hamid and Balamurugan in \cite{hamid}. A dominating set $S$ is an \emph{Isolate Dominating Set} (\textit{IDS}) if the induced subgraph $G[S]$ has at least one isolated vertex. The \textit{isolate domination number} $\gamma_0(G)$ is the minimum size of an IDS of $G$. Clearly, every independent dominating set in a graph is an isolate dominating set, so every graph possess an isolate dominating set. Given a graph $ G=(V,E)$ and a positive integer $ k,$ the IDOM problem is to check whether $ G $ has an isolate dominating set of size at most $ k.$ In \cite{rad}, N.J. Rad has proved that the IDOM problem is NP-complete, even when restricted to bipartite graphs. In this paper, we extended this result by proving that this problem is NP-complete for even split graphs and perfect elimination bipartite graphs. We initiated the study of new variant of domination called \textit{independent secure domination} with the following definition.\\
A set $S\subseteq V$ is an \textit{Independent Secure Dominating Set} (\textit{InSDS}) if $S$ is an independent set and a SDS of $G$. The minimum size of an InSDS in $G$ is called the \textit{independent secure domination number} of $G$ and is denoted by $\gamma_{is}(G)$. \\[4pt]
Given a graph $ G$ and a positive integer $ k,$ the InSDM problem is to check whether $ G $ has an independent secure dominating set of size at most $ k.$ The MInSDS problem is to find an independent secure dominating set of minimum size, in the input graph.
\paragraph{Motivation} A communication network is modeled as a graph $ G=(V, E) $ where each node represents a communicating device and each edge represents a communication link between two devices. All the nodes in the network need to communicate and exchange information to perform a task. However, in the networks where the reliability of the nodes is not guaranteed,  every node $ v $ can be a potential malicious node. A virtual protection device at node $ v $ can (i) detect the malicious activity at any node $ u $ in its closed neighborhood (ii) migrate to the malicious node $ u $ and repair it. One is interested to deploy minimum number of virtual protection devices such that every node should have at least one virtual protection device within one hop distance even after virtual protection  device is migrated to malicious node. This problem can be solved by finding a minimum secure dominating set of the graph $ G $. Further, if two virtual protection devices are deployed adjacently then there is a chance of one virtual protection device getting damaged or corrupted if other one is corrupted. To avoid this, the virtual protection devices should be deployed on the nodes such that their neighborhood should not contain any other virtual protection devices. This can be solved by finding a minimum independent secure dominating set of the graph $ G $. \par 
The rest of the paper is organized as follows. In Section 2, some basic results related to independent secure domination are presented. In Section 3, algorithmic complexity of IDOM problem is investigated. In Section 4, the InSDM problem is proved as NP-complete even when restricted to bipartite graphs. Also, the computational complexity difference of InSDM problem with DOMINATION problem is highlighted and some approximation results related to MInSDS problem are presented. 
\section{Basic results}
In this section, some precise values and bounds for independent secure domination in frequently encountered graph classes are presented. 
The following observation is immediate from the definition.
\begin{Observation}\label{2}
	For a graph $ G $, $ \gamma_{is}(G) \ge \gamma_{s}(G).$
\end{Observation}
\begin{prop}\label{1}
	For the complete graph $ K_n $ with $ n $ vertices, $\gamma_{is}(K_n)=1.$
\end{prop}
\begin{thm}
	Let $ G $ be a graph with $ n $ vertices. Then $ \gamma_{is}(G) =1$ if and only if $ G = K_n $.
\end{thm}
\begin{proof}
	Suppose $ \gamma_{is}(G) = 1 $ and let $ S = \{v\} $ be a minimum size InSDS of $ G $. Suppose $G \ne K_n$, then there exists $x, y \in V (G)$ such that $d(x, y) \ge 2$. Then $(S \setminus \{v\}) \cup \{x\} = \{x\}$, which is not a dominating set of $ G $, since $ (x, y) \notin E $. Therefore,
	$ G = K_n $. The converse is true by proposition \ref{1}.
\end{proof}
\begin{prop}
	For the complete bipartite graph $ K_{p,q} $ with $p \le q $ vertices, 	
	$$\gamma_{is}(K_{p,q}) = 
	\begin{cases}
	q, & 
	\text{if } p=1 \\
	p, & \text{otherwise}  
	\end{cases}
	$$
\end{prop}	
\begin{proof}
	Suppose  $ G=(A, B, E)=K_{p,q} $, where $ \vert A \vert =p $  and $ \vert B \vert =q $ be a complete bipartite graph. If $ p=1 $, it is clear that an InSDS can be formed with all the vertices of $ B $. Therefore, $\gamma_{is}(G) \le q.$ There is no independent secure dominating set possible, which contains a vertex from $ A.$ Hence $\gamma_{is}(G) \ge q.$ It can be observed that if $ p\ge 2 $, then an InSDS can be formed with all the vertices of $ A $. Therefore, $\gamma_{is}(G) \le p.$ There is no independent secure dominating set possible, which contains a vertex from each partite set $A$ and $ B $. Hence $\gamma_{is}(G) \ge p.$ 
\end{proof}
\begin{prop}\label{3}
	Let $ P_n $ be a path graph with $ n$ $(\ge 4)$ vertices. Then $\gamma_{is}(P_n)=\lceil \frac{3n}{7} \rceil$.
\end{prop}
\begin{proof}
	From \cite{pog}, we know that $ \gamma_s(P_n)= \lceil \frac{3n}{7} \rceil.$ By this and proposition \ref{2}, we have $\gamma_{is}(P_n)$ is at least $\lceil \frac{3n}{7} \rceil.$ Hence in order to complete the proof, we need to exhibit an ISDS and InSDS of $ P_n $ of size $\lceil \frac{3n}{7} \rceil$. 
	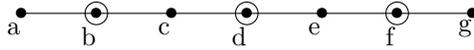
\begin{figure}[H]
		\begin{center}
			
			\begin{tikzpicture}
			\draw (0,0)--(6,0);
			\draw (1,0) circle (1.5mm);
			\draw (3,0) circle (1.5mm);
			\draw (5,0) circle (1.5mm);
			
			\node at (0,0){\textbullet};	\node at (-0.1,-0.2){a};
			\node at (1,0){\textbullet};	\node at (0.9,-0.3){b};
			\node at (2,0){\textbullet};	\node at (1.9,-0.2){c};
			\node at (3,0){\textbullet};	\node at (2.9,-0.3){d};
			\node at (4,0){\textbullet};	\node at (3.9,-0.2){e};
			\node at (5,0){\textbullet};	\node at (4.9,-0.3){f};
			\node at (6,0){\textbullet};	\node at (5.9,-0.2){g};
			
			\end{tikzpicture}
		\end{center}
		\caption{ Path graph $ P_{7} $}
	\end{figure}
	Suppose $ P_n=(v_1,v_2,\ldots,v_n) $ and $ n=7m+r, $ where $ m\ge 1 $ and $ 0 \le r \le 6.$ Define two sets \[
	X = \bigcup_{i=0}^{m-1}\{v_{7i+2},v_{7i+4},v_{7i+6}\}
	\] and \[Y=
	\begin{cases}
	\emptyset, & 
	\text{if } r=0 \\
	\{v_{7m+1}\}, & \text{if } r=1 \text{ or } 2 \\
	\{v_{7m+1},v_{7m+3}\}, & \text{if } r=3 \text{ or }4 \\
	\{v_{7m+1},v_{7m+3},v_{7m+5}\}, & \text{if } r=5 \text{ or } 6 \\  
	\end{cases}
	\]
	Clearly $ \vert X \cup Y \vert = \lceil \frac{3n}{7} \rceil $. It can be noted that the set $ X \cup Y $ forms an independent set. It can also be verified that $ X \cup Y $ forms InSDS of $ P_n $. Hence the result.   
\end{proof}
\begin{Observation}\label{pathcycle}
	If $ G^\prime $ is a spanning subgraph of graph $ G $, then $ \gamma_{is}(G^\prime) \ge \gamma_{is}(G). $
\end{Observation}
\begin{prop}\label{cycle}
	If $ C_n $ is a cycle graph with $ n $ vertices, then $ \gamma_{is}(C_n)=\lceil \frac{3n}{7} \rceil (n\ne 3, 5)$.
\end{prop}
\begin{proof}
	From observation \ref{pathcycle}, $ \gamma_{is}(P_n) \ge \gamma_{is}(C_n)$.  From proposition \ref{3}, it can be noted that $ \gamma_{is}(C_n) \le \lceil \frac{3n}{7} \rceil $. InSDS does not exist for $ C_5 $. From observation \ref{2}, we know that $ \gamma_{s}(G) \le \gamma_{is}(G).$ It is known that $ \gamma_s(C_n)= \lceil \frac{3n}{7} \rceil$ \cite{pog}. Hence the result.  
\end{proof}
\begin{thm}
	Let $ G $ be a non-complete graph. Then $ \gamma_{is}(K_1+G)=\gamma_{is}(G). $
\end{thm}
\begin{proof}
	Let $ S $ be an InSDS of a graph $ G $. Note that $ S$ is also an InSDS of $ K_1+G$ i.e., $ \gamma_{is}(K_1+G)\le \gamma_{is}(G).$ Let $ v $ be the vertex of $K_1$. Since $ K_1+G $ is not a complete graph, $ S $ should contain at least two vertices. If $ v \in S $ then no other vertex from $ K_1+G $ can be included in $ S $. Therefore $ v \notin S$ and hence $\gamma_{is}(K_1+G) \ge \gamma_{is}(G).$
\end{proof}
\noindent The following result can be obtained from above theorem and proposition \ref{cycle}.
\begin{remark}
	Let $ W_n $ be a wheel graph with $ n+1$ $(\ne 6)$ vertices then, $ \gamma_{is}(W_n)=\lceil \frac{3n}{7} \rceil. $
\end{remark}
\begin{thm}
	Let $ G= P_m \Box P_k $ be the Cartesian product of two path graphs $ P_m $ and $ P_k $. Then, $ \gamma_{is}(G) \le \lceil \frac{mk}{3}\rceil+4 $.
\end{thm}
\begin{proof}
	We construct an InSDS of the grid graph $ G=P_m \Box P_k $ of size at most $ \lceil \frac{mk}{3}\rceil+4 $ by considering standard plane embedding as depicted in figure \ref{fig:assgn}. Include all solid vertices in a set $ X $ and let $ Y=V \setminus X. $ It is clear that $ \vert X \vert =  \lceil\frac{mk}{3}\rceil.$ 	
	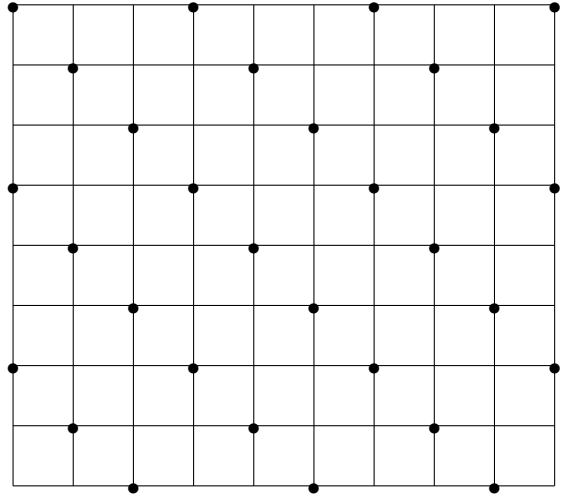
\begin{figure}
		\begin{center}
			\begin{tikzpicture}
			\draw (0,0)--(7.2,0);
			\draw (0,0.8)--(7.2,0.8);
			\draw (0,1.6)--(7.2,1.6);
			\draw (0,2.4)--(7.2,2.4);
			\draw (0,3.2)--(7.2,3.2);
			\draw (0,4)--(7.2,4);
			\draw (0,4.8)--(7.2,4.8);
			\draw (0,5.6)--(7.2,5.6);
			\draw (0,6.4)--(7.2,6.4);
			
			\node at (0,6.35){\textbullet};			\node at (2.4,6.35){\textbullet};			\node at (4.8,6.35){\textbullet};			\node at (7.2,6.35){\textbullet};
			
			\node at (0.8,5.55){\textbullet};		\node at (3.2,5.55){\textbullet};			\node at (5.6,5.55){\textbullet};
			
			\node at (1.6,4.75){\textbullet};		\node at (4,4.75){\textbullet};				\node at (6.4,4.75){\textbullet};
			
			\node at (0,3.95){\textbullet};			\node at (2.4,3.95){\textbullet};			\node at (4.8,3.95){\textbullet};			\node at (7.2,3.95){\textbullet};
			
			\node at (0.8,3.15){\textbullet};		\node at (3.2,3.15){\textbullet};			\node at (5.6,3.15){\textbullet};
			
			\node at (1.6,2.35){\textbullet};		\node at (4,2.35){\textbullet};				\node at (6.4,2.35){\textbullet};
			
			\node at (0,1.55){\textbullet};			\node at (2.4,1.55){\textbullet};			\node at (4.8,1.55){\textbullet};			\node at (7.2,1.55){\textbullet};
			
			\node at (0.8,0.75){\textbullet};		\node at (3.2,0.75){\textbullet};			\node at (5.6,0.75){\textbullet};
			
			\node at (1.6,-0.05){\textbullet};		\node at (4,-0.05){\textbullet};			\node at (6.4,-0.05){\textbullet};
			
			\draw (7.2,6.4)--(7.2,0);			\draw (6.4,6.4)--(6.4,0);			\draw (5.6,6.4)--(5.6,0);			\draw (4.8,6.4)--(4.8,0);
			\draw (4,6.4)--(4,0);				\draw (3.2,6.4)--(3.2,0);			\draw (2.4,6.4)--(2.4,0);			\draw (1.6,6.4)--(1.6,0);
			\draw (0.8,6.4)--(0.8,0);			\draw (0,6.4)--(0,0);
			\end{tikzpicture}
			\caption{Grid graph $ P_9 \Box P_{10} $}
			\label{fig:assgn}
		\end{center}
	\end{figure}
	Depending upon the value of $ m$ and $	k$ $modulo$ $3$, three cases are possible. In every case, each $ u \in Y $ which is not within distance $ 3 $ of top right corner or bottom left corner, is defended by a vertex in $ X $. Therefore, in each case add triangle vertices to $ X $ by removing square vertices from it as illustrated in figure \ref{fig12}. Now, all the vertices of $ G $ can be defended by $ X.$ Hence, $ X $ forms an InSDS of size at most $\lceil\frac{mk}{3}\rceil+4 $.
	\begin{figure}
		\begin{center}
			
			\begin{tikzpicture}
			
			\draw (0,2.4) --(3.2,2.4);
			\draw (.8,1.6) --(3.2,1.6);
			\draw (1.6,.8) --(3.2,.8);
			\draw (2.4,0) --(3.2,0);
			
			\draw (0.8,2.4) --(0.8,1.6);
			\draw (1.6,2.4) --(1.6,0.8);
			\draw (2.4,2.4) --(2.4,0);
			\draw (3.2,2.4) --(3.2,-0.8);
			
			\node at (1.6,2.4){\textbullet};
			\node at (2.4,1.6){\textbullet};
			\node at (3.2,0.8){\textbullet};
			
			\node at (1.6,-0.8) {$ (a) $};
			\draw (3,2.2)--(3.4,2.2) -- (3.2,2.6)-- (3,2.2);

			\draw (4.4,2.4) --(7.6,2.4);
			\draw (5.2,1.6) --(7.6,1.6);
			\draw (6,.8) --(7.6,.8);
			\draw (6.8,0) --(7.6,0);
			
			\draw (5.2,2.4) --(5.2,1.6);
			\draw (6,2.4) --(6,0.8);
			\draw (6.8,2.4) --(6.8,0);
			\draw (7.6,2.4) --(7.6,-0.8);
			
			\node at (5.2,2.4){\textbullet};
			\node at (6,1.6){\textbullet};
			\node at (6.8,0.8){\textbullet};
			\node at (7.6,0){\textbullet};
			
			\filldraw (7.6,2.4) circle (2pt);
			\draw (7.4,2.2)--(7.8,2.2)--(7.8,2.6)--(7.4,2.6)--(7.4,2.2);
			\draw (6.6,2.2)--(7,2.2) -- (6.8,2.6)-- (6.6,2.2);
			\draw (7.4,1.4)--(7.8,1.4) -- (7.6,1.8)-- (7.4,1.4);
			
			\node at (6,-0.8) {$ (b) $};
			
			\draw (8.8,2.4) --(12,2.4);
			\draw (9.6,1.6) --(12,1.6);
			\draw (10.4,.8) --(12,.8);
			\draw (11.2,0) --(12,0);
			
			\draw (9.6,2.4) --(9.6,1.6);
			\draw (10.4,2.4) --(10.4,0.8);
			\draw (11.2,2.4) --(11.2,0);
			\draw (12,2.4) --(12,-0.8);
			
			\node at (8.8,2.4){\textbullet};
			\node at (9.6,1.6){\textbullet};
			\node at (10.4,0.8){\textbullet};
			\node at (11.2,0){\textbullet};
			\node at (12,-0.8){\textbullet};
			
			\node at (11.2,2.4){\textbullet};
			\filldraw (11.2,2.4) circle (2pt);
			\draw (11,2.2)--(11.4,2.2)--(11.4,2.6)--(11,2.6)--(11,2.2);
			\node at (12,1.6){\textbullet};
			\filldraw (12,1.6) circle (2pt);
			\draw (11.8,1.4)--(12.2,1.4)--(12.2,1.8)--(11.8,1.8)--(11.8,1.4);
			
			\draw (10.2,2.2)--(10.6,2.2) -- (10.4,2.6)-- (10.2,2.2);
			\draw (11,1.4)--(11.4,1.4) -- (11.2,1.8)-- (11,1.4);
			\draw (11.8,0.6)--(12.2,0.6) -- (12,1)-- (11.8,0.6);
			\draw (11.8,2.2)--(12.2,2.2) -- (12,2.6)-- (11.8,2.2);
			
			\node at (10.4,-0.8) {$ (c) $};
			\end{tikzpicture}
			\caption{Three possible cases for top right corners of $ P_m \Box P_k $} \label{fig12}
		\end{center}
	\end{figure}
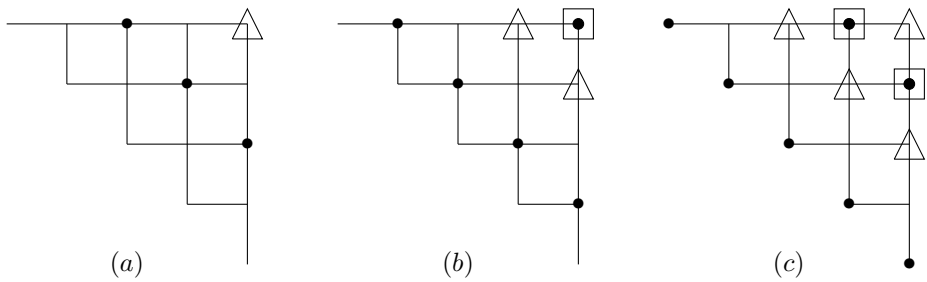
\end{proof}
\section{Complexity of isolate domination}
\label{sec:1}
\noindent In this section, we show that the IDOM problem is NP-complete for split graphs and perfect elimination bipartite graphs. \\[4pt]
\noindent 
\noindent The SET COVER DECISION problem which is used to prove NP-completeness of IDOM for split graphs, is defined as follows.\\[4pt]
\noindent
\textit{SET COVER DECISION problem} (\textit{SET-COVER}) \\ [6pt]
\textit{Instance:} A finite set $X$ of elements, a family of $m$ subsets of elements $C$, and a positive integer $k$.\\
\textit{Question:} Does there exist a subfamily of $k$ subsets $C^\prime$ whose union equals $X$? \\[5pt] 
SET-COVER problem has been proved as NP-complete by R.M. Karp \cite{karp}.
\begin{thm}
	IDOM is NP-complete for split graphs.
\end{thm}
\begin{proof} 	
	Given a split graph $G=(V,E)$, a positive integer $k$ and an arbitrary set $D \subseteq V$, we can check in polynomial time whether $D$ is an IDS of $G$ of size at most $ k$. Hence IDOM problem for split graphs is in NP. \par
	To prove NP-hardness of IDOM for split graphs, we propose a reduction from SET-COVER problem. Let $X=\{x_1$, $x_2,\ldots,x_n\}$, $C=\{C_1,C_2,\ldots,$ $C_m\}$, $\vert X \vert =n$ and $\vert C \vert = m$ be an instance of SET-COVER problem. Construct a graph $G$ by creating the following vertices: (i) a vertex $x_i$ for each element $x_i \in X$; (ii) vertices $c_j,u_j$ and $v_j$ for each subset $C_j \in C$. Let $I= \{x_1,x_2,\ldots,x_n\}$, $J=\{v_1,v_2,\ldots,v_m\}$, $K=\{c_1,c_2,\ldots,c_m\}$ and $L=\{u_1,u_2,\ldots,u_m\}$. Add the following edges in $G$: (i) if $x_i \in C_j$, then add edge $(x_i,c_j)$, where $1\le i \le n$ and $1 \le j \le m$.
	(ii) edges between every pair of vertices in the set $K \cup L$. (iii) edges between $u_i$ and $v_i$ where $1\le i \le m$. It can be observed that the set $I \cup J$ is an independent set and $ K \cup L $ is a clique in $G$. Therefore, $G$ is a split graph and can be constructed in polynomial time. An example construction of graph $G$ with SET-COVER instance with $X=\{x_1,x_2$,$\ldots,x_9\}$ and $C=\{\{x_1,x_2,x_4\},\{x_1,x_2,x_3\},\{x_3,x_6,x_9\},\{x_4,x_5,x_6,x_9\},$ $\{x_5,x_7,x_8\}\}$ is shown in figure \ref{setc}.\par 
	\begin{figure}
		\begin{center}
			
			\begin{tikzpicture}[scale=1.20]
			\node at (0.5,0){\textbullet};\node at (0.5,-0.3){$ v_1 $};
			\node at (0.5,1){\textbullet};\node at (0.5,1.2){$ u_1 $};
			
			\draw (0.5,0)--(0.5,1);
			\node at (0.5,2){\textbullet};\node at (0.5,1.7){$ c_1 $};
			
			\draw (0.5,2)--(0,3)--(2.5,2)--(2,3)--(4.5,2);
			\draw (7.75,2)--(4,3)--(6.5,2)--(5,3)--(4.5,2);
			\draw (1,3)--(0.5,2)--(3,3)--(6.5,2);
			\draw (2.5,2)--(1,3);
			\draw (7,3)--(7.75,2)--(6,3);
			\draw (8,3)--(4.5,2);
			
			\node at (0,3){\textbullet};\node at (0,3.3){$ x_1 $};
			
			\node at (2.5,0){\textbullet};\node at (2.5,-0.3){$ v_2 $};
			\draw (2.5,0)--(2.5,1);
			\node at (1,3){\textbullet};\node at (1,3.3){$ x_2 $};
			
			\node at (4.5,0){\textbullet};\node at (4.5,-0.3){$ v_3 $};
			\node at (2.5,1){\textbullet};\node at (2.5,1.2){$ u_2 $};
			\draw (4.5,0)--(4.5,1);
			\draw (6.5,0)--(6.5,1);
			\draw (7.75,0)--(7.75,1);
			\node at (2.5,2){\textbullet};\node at (2.5,1.7){$ c_2 $};
			\node at (2,3){\textbullet};\node at (2,3.3){$ x_3 $};
			
			\node at (6.5,0){\textbullet};\node at (6.5,-0.3){$ v_4 $};
			\node at (3,3){\textbullet};\node at (3,3.3){$ x_4 $};
			
			\node at (7.75,0){\textbullet};\node at (7.75,-0.3){$ v_5 $};
			\node at (4.5,1){\textbullet};\node at (4.5,1.2){$ u_3 $};
			\node at (4.5,2){\textbullet};\node at (4.5,1.7){$ c_3 $};
			\node at (4,3){\textbullet};\node at (4,3.3){$ x_5 $};
			
			\node at (5,3){\textbullet};\node at (5,3.3){$ x_6 $};
			
			\node at (6.5,1){\textbullet};\node at (6.5,1.2){$ u_4 $};
			\node at (6.5,2){\textbullet};\node at (6.5,1.7){$ c_4 $};
			\node at (6,3){\textbullet};\node at (6,3.3){$ x_7 $};
			
			\node at (7,3){\textbullet};\node at (7,3.3){$ x_8 $};
			
			\node at (7.75,1){\textbullet};\node at (7.75,1.2){$ u_5 $};
			\node at (7.75,2){\textbullet};\node at (7.75,1.7){$ c_5 $};
			\node at (8,3){\textbullet};\node at (8,3.3){$ x_9 $};
			
			\draw (8,3)--(6.5,2);
			
			\draw (0,0.5)--(8,0.5)--(8,2.5)--(0,2.5)--(0,0.5);
			
			\end{tikzpicture}
			\caption{Construction of a split graph $G$ from an instance of SET-COVER}
			\label{setc}
		\end{center}
	\end{figure}
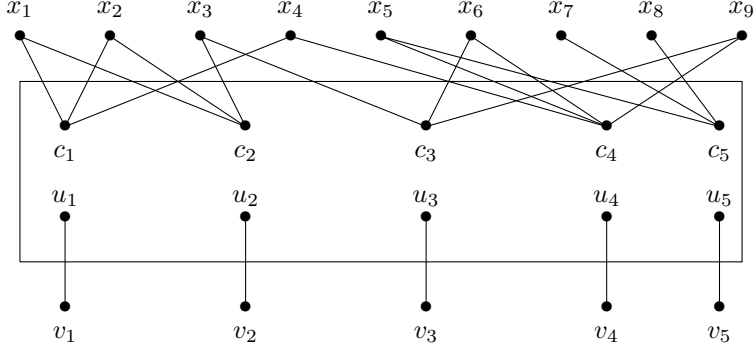
	Now we show that given instance of SET-COVER problem $<X,C>$ has a set cover of size at most $k$ if and only if the constructed graph $G$ has an IDS of size at most $m+k$. Suppose $C^\prime \subseteq C$ is a set cover of $X$, with $\vert C^\prime \vert = k$, then it is easy to verify that the set $D= \{c_j : C_j \in C^\prime\} \cup \{v_i$ $:$ $1 \le i \le m\}$ is an IDS of size at most $m+k$ in $G$.\par
	Conversely, suppose $D \subseteq V$, be an IDS of size at most $m+k$ in $G$. It can be easily seen that every IDS of $G$ must include at least one of $u_j$ or $v_j$ for every $j$, where $1 \le j \le m$. Since the vertex $v_j$ is a pendant vertex, $\vert D \cap (J \cup L) \vert \ge m$ which implies $\vert D \cap (I \cup K) \vert \le k$. Now we have two possible cases. \textit{Case }(\textit{i}) $ D \cap I = \phi$, it can be easily verified that $D \cap K$ is a set cover of size at most $k$. \textit{Case} (\textit{ii}) $ D \cap I \ne \phi$, since $I$ is an independent set, every vertex in $D \cap I$ can be replaced with its adjacent vertex in the set $S$. Hence there exists a set cover of size at most $ k. $  
\end{proof}
\noindent The decision version of domination problem is defined as follows.\\[4pt] 
\noindent 
\textit{DOMINATION DECISION problem (DOM)} \\ [6pt]
\textit{Instance:} A simple,undirected graph $G = (V, E)$ and a positive integer $k$.\\
\textit{Question:} Does there exist a dominating set of size at most $ k $ in $ G $ ?\\[4pt]
\noindent The DOM problem in bipartite graphs (DSDPB) has been proved as NP-complete by A.A. Bertossi \cite{bersto}.
\begin{thm}
	IDOM is NP-complete for perfect elimination bipartite graphs.
\end{thm}
\begin{proof}
	It is known that IDOM is in NP. To prove the NP-hardness, we provide a polynomial transformation from an instance of DSDPB to an instance of IDOM problem in perfect elimination bipartite graphs. Given a bipartite graph $G=(X,Y,E)$, we construct a graph $G^\prime=(X^\prime,Y^\prime,E^\prime)$ as follows. Let $Y_1=\{y \in Y : x \in N(y), d(x) \ne 1 \}$, $Y_2=Y \setminus Y_1$, i.e. $Y_1$ is the set of vertices in $Y$ which have degree one neighbors. With out of loss of generality, let for some $l \ge 0$, $Y_1=\{ y_1, y_2,\ldots,y_l\}$ and for some $q \ge 0$, $Y_2=\{ y_{l+1}, y_{l+2},\ldots,y_q\}$.
	Let $X^\prime = X \cup (\cup_{i=1}^l \{a_i, c_i\})$, $Y^\prime= Y \cup (\cup_{i=1}^l \{b_i\})$ and $E^\prime=E \cup (\cup_{i=1}^l \{(y_i,a_i),(a_i,b_i),(b_i,c_i)\})$. For all $1 \le i \le l$, let $A=\{a_1, a_2, \ldots, a_l\}, B=\{b_1, b_2,\ldots, b_l\},$ and $C=\{c_1, c_2,\ldots, c_l\}$. Here $G^\prime=(X^\prime,Y^\prime,E^\prime)$ is obtained from $G$ by attaching a path $P_3$ to $y_i$ for all $i$, $1 \le i \le l$. An example construction of a graph $G^\prime$ from a graph $G$ is shown in figure \ref{fig:test}.  \par
	Note that $G^\prime$ is a bipartite graph and can be constructed from $G$ in polynomial time. It can be verified that ordering $\sigma=(b_1c_1,b_2c_2,\ldots,b_lc_l,a_1y_1,a_2y_2$, $\ldots,a_ly_l,x_{l+1}y_{l+1}$,
	$x_{l+2}y_{l+2},\ldots$, $x_qy_q)$ is a perfect edge elimination ordering of $G^\prime$. Therefore $G^\prime$ is a perfect elimination bipartite graph. \par
	\begin{figure}
		\begin{center}
			\begin{tikzpicture}[scale=1.20]
			\node at (0,0){\textbullet}; \node at (0,-0.2){$ x_5 $};
			\node at (0,1){\textbullet}; \node at (0,0.8){$ x_4 $};
			\node at (0,2){\textbullet}; \node at (0,1.8){$ x_3 $};
			\node at (0,3){\textbullet}; \node at (0,2.8){$ x_2 $};
			\node at (0,4){\textbullet}; \node at (0,3.8){$ x_1 $};
			
			\node at (4,0){\textbullet}; \node at (4,-0.2){$ y_5 $};
			\node at (4,1){\textbullet}; \node at (4,0.8){$ y_4 $};
			\node at (4,2){\textbullet}; \node at (4,1.8){$ y_3 $};
			\node at (4,3){\textbullet}; \node at (4,2.8){$ y_2 $};
			\node at (4,4){\textbullet}; \node at (4,3.8){$ y_1 $};
			
			\draw (0,4)--(4,4)--(0,3)--(4,3)--(0,4);
			\draw (0,4)--(4,4)--(0,3)--(4,3)--(0,4);
			\draw (0,2)--(4,1)--(0,1);
			\draw (0,2)--(4,3);
			\draw (0,2)--(4,0);
			\draw (4,2)--(0,3)--(4,1);
			\draw (0,0)--(4,0);
			
			\node at (5.5,2){\textbullet}; \node at (5.5,1.8){$ a_3 $};
			\node at (5.5,3){\textbullet}; \node at (5.5,2.8){$ a_2 $};
			\node at (5.5,4){\textbullet}; \node at (5.5,3.8){$ a_1 $};
			
			\node at (6.75,2){\textbullet}; \node at (6.75,1.8){$ b_3 $};
			\node at (6.75,3){\textbullet}; \node at (6.75,2.8){$ b_2 $};
			\node at (6.75,4){\textbullet}; \node at (6.75,3.8){$ b_1 $};
			
			\node at (8,2){\textbullet}; \node at (8,1.8){$ c_3 $};
			\node at (8,3){\textbullet}; \node at (8,2.8){$ c_2 $};
			\node at (8,4){\textbullet}; \node at (8,3.8){$ c_1 $};
			
			\draw (4,4)--(6,4)--(8,4);
			\draw (4,3)--(6,3)--(8,3);
			\draw (4,2)--(6,2)--(8,2);
			
			\node at (2.2,4.35){$ G $};
			\draw[dotted] (4.3,-0.4)--(4.3,4.6)--(-0.2,4.6)--(-0.2,-0.4)--(4.3,-0.4);		
			\end{tikzpicture}
			\caption{Transformation from graph $G$ to graph $G^\prime$}
			\label{fig:test}
		\end{center}
	\end{figure}
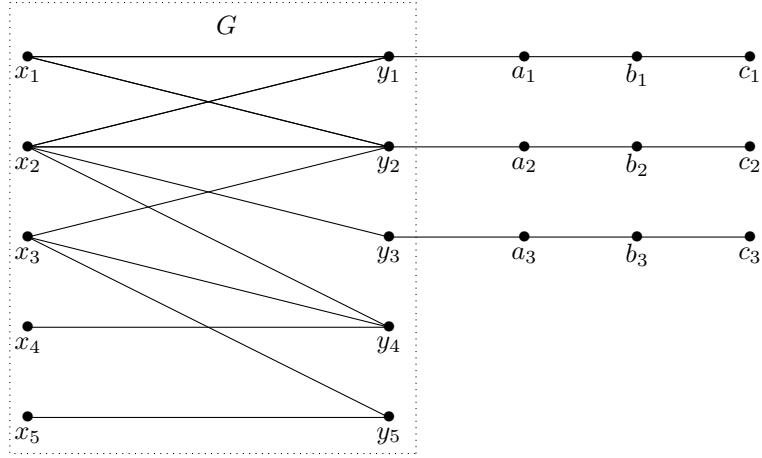
	Now we show that $G$ has a dominating set of size at most $k$ if and only if $G^\prime$ has an IDS of size at most $k+l$. Suppose $D$ be a dominating set of size at most $k$ in $G$, then it can be observed that $D \cup \{b_1,b_2,\ldots,b_l\}$ is an IDS of size at most $k+l$ in $G^\prime$. Conversely, suppose that $D^\prime$ is an IDS of size at most $p=k+l$ in $G^\prime$. It can be verified that $\vert D^\prime \cap \{b_i,c_i\} \vert \ge 1$, for all $1 \le i \le l$ and hence $\vert D^\prime \cap (B \cup C) \vert \ge l$. It can be easily seen that  $(D^\prime \setminus (A \cup B \cup C)) \cup \{y_i$ $:$ $y_i \notin D^\prime$ and  $ a_i \in D^\prime\}$ is a dominating set of size at most $k$ in $G$. 
\end{proof}

\section{Independent secure domination}
A set $S\subseteq V$ is an \textit{Independent Secure Dominating Set} (\textit{InSDS}) if $S$ is an independent set and a SDS of $G$.
In this section, the computational complexity of InSDM problem has been investigated for several graph classes. Some approximation results of this parameter also presented. 
The decision version of independent secure domination and independent domination problems are defined as follows.\\[3pt] 
\textit{INDEPENDENT SECURE DOMINATION DECISION problem (InSDM)} \\ [6pt]
\textit{Instance:} A simple, undirected graph $G=(V, E)$ and a positive integer $k$.\\
\textit{Question:} Does there exist an InSDS of size at most $ k $ in $ G $ ?\\[3pt]
\noindent
\textit{INDEPENDENT DOMINATION DECISION problem (InDM)} \\ [6pt]
\textit{Instance:} A simple, undirected graph $G=(V, E)$ and a positive integer $l$.\\
\textit{Question:} Does there exist an InDS of size at most $ l $ in $ G $ ?\\[6pt]
Garey and Johnson \cite{garey} have proved that InDM as NP-complete.
\begin{thm}\label{insdsnp}
	InSDM is NP-complete.
\end{thm}
\begin{proof} It is easy to verify a yes instance of InSDM in polynomial time. Hence InSDM is in NP. To prove NP-hardness of InSDM, we propose a reduction from InDM as follows. Given an instance $G=(V,E)$ of InDM, with $ V = \{v_1,v_2,\ldots,v_n\} $, we get an instance of InSDM by constructing a graph $G^\prime=(V^\prime, E^\prime)$ where $V^\prime=$ $V$ $\cup$ $\{a_i, b_i$ $:$ $v_i \in V\}$ and $E^\prime=E$ $\cup$ $\{(v_i,a_i),(a_i,b_i)$ $:$ $v_i \in V\}.$ 
	Clearly, $G^\prime$ can be constructed from $G$ in polynomial time. \par
	Next, we shall show that $G$ has an InDS of size at most $k$ if and only if $G^\prime$ has an InSDS of size at most $p=n+k$. Let $D$ be an InDS of size at most $k$ in $G$, then $D \cup \lbrace a_i : v_i \notin D \rbrace \cup  \lbrace b_i : v_i \in D \rbrace$ is an InSDS of size at most $n+k$ in $G^\prime$.\par
	Conversely, suppose $D^*$ is an InSDS of $G^\prime$, with $\vert D^* \vert \le n+k$. Let $A=\{a_i : v_i \in V\}$, $B=\{b_i : v_i \in V\}$ and $D^\prime= D^* \cap (A \cup B)$. Since the set $D^*$ is an InSDS of $G^\prime$, $\vert D^* \cap A \vert < n$ and $\vert D^* \cap B \vert < n$. Since every InSDS should contain either $a_i$ or $b_i$ for every $v_i \in V$ it follows that $\vert D^\prime \vert = n$. Therefore, no vertex in $D^\prime$ can defend any vertex $v_i \in V^\prime \setminus(A \cup B)$ and hence the set $D^* \setminus D^\prime$ is an InDS of size at most $k$ in $G$. 
\end{proof}
\noindent Note that in theorem \ref{insdsnp}, if the graph $ G $ is bipartite then the constructed graph $ G^\prime $ is also bipartite. Since InDM problem is NP-complete for bipartite graphs \cite{indomb}, it can imply the following theorem.
\begin{thm}
	InSDM is NP-complete for bipartite graphs.
\end{thm}
\subsection{Independent secure domination for bounded tree-width graphs}
Let $ G $ be a graph, $ T $ be a tree and $v\ $ be a family of vertex sets $ V_t \subseteq V (G) $ indexed by the vertices $ t $ of $ T $ . 
The pair $ (T, v\ ) $ is called a tree-decomposition of $ G $ if it satisfies the following three conditions:
(i) $ V(G)  =  \bigcup_{t \in V(T)} V_t $,
(ii) for every edge $ e \in E(G) $ there exists a $ t \in V(T)$ such that both ends of $ e $ lie in $ V_t $,
(iii) $V_{t_1} \cap   V_{t_3}  \subseteq V_{t_2}$ whenever $ t_1$, $t_2$, $t_3 \in V(T) $ and $ t_2 $ is on the path in $ T $ from $ t_1 $ to $ t_3 $.
The width of $ (T, v\ ) $ is the number $ max\{\vert V_t \vert-1 : t\in T \}$, and the tree-width $ tw(G) $ of $ G $ is the minimum width of any tree-decomposition of $ G $. By Courcelle's Thoerem, it is well known that every graph problem that can be described by counting monadic second-order logic (CMSOL) can be solved in linear-time in graphs of bounded tree-width, given a tree decomposition as input \cite{courc}. We show that InSDM problem can be expressed in CMSOL. 
\begin{thm}[\textit{Courcelle's Theorem}](\cite{courc})\label{cmsol1}
	Let $ P $ be a graph property expressible in CMSOL and let $ k $ be
	a constant. Then, for any graph $ G $ of tree-width at most $ k $, it can be checked in
	linear-time whether $ G $ has property $ P $.
\end{thm}
\begin{thm}\label{cmsol2}
	Given a graph $ G $ and a positive integer $ k $, InSDM can be expressed in CMSOL.
\end{thm}
\begin{proof}
	First, we present the CMSOL formula which expresses that the graph $ G $ has a dominating set of size at most $ k.$
	{\small 	$$Dominating(S)= (\vert S \vert \le k)  \land (\forall p)((\exists q)(q\in S \land adj(p,q))) \lor (p\in S)$$}
	where $ adj(p, q) $ is the binary adjacency relation which holds if and only if, $ p, q $ are two adjacent vertices of $ G.$
	$ Dominating(S) $ ensures that for every vertex $ p \in V $, either $ p\in S $ or $ p $ is adjacent to a vertex in $ S$ and the cardinality of $ S $ is at most $ k.$ 
	
	The set $ S \subseteq V$ is independent if and only if there does not exist a partition of $ S $ into two sets $ S_1 $ and $ S_2 $ such that there is an edge between a vertex in $ S_1 $ and a vertex in $ S_2 $. The CMSOL formula to express that the set $ S $ is independent as follows.
	{\small $$Independent(S)= \neg (\exists S_1, S_1 \subseteq S, (\exists e \in E, \exists u \in S_1, \exists v \in S \setminus S_1, (inc(u,e) \land inc(v,e)) ))$$}
	where $ inc(v, e) $ is the binary incidence relation which hold if and only if edge $ e $ is incident to vertex $ v $ in $ G.$
	Now, by using the above two CMSOL formulas we can express InSDM in CMSOL formula as follows.\\
	{\small 	$ InSDM(S)=Dominating(S) \land Independent(S) \land (\forall x)((x \in S) \lor$} \\{\small \hspace*{2cm}$((\exists y) (y \in S \land adj(x,y) \land Dominating((S \setminus \{y\}) \cup \{x\}) )))  $}\\
	Therefore, InSDM can be expressed in CMSOL.
\end{proof}
\noindent Now, the following result is immediate from Theorems \ref{cmsol1} and \ref{cmsol2}.
\begin{thm}
	InSDM can be solvable in linear time for bounded tree-width graphs. 
\end{thm}
\subsection{Independent secure domination for threshold graphs}
\noindent Threshold graphs have been studied with the following definition \cite{threshold1}.
\begin{definition}
	A graph $ G=(V, E) $ is called a \textit{threshold graph} if there is a real number $ T $ and a real number $ w(v) $ for every $ v \in V $ such that a set $ S \subseteq V $ is independent if and only if $ \sum_{v \in S}w(S) \le T $. 
\end{definition}
\noindent Although several characterizations defined for threshold graphs, the following is useful in solving the InSDM problem.\\
A graph $ G $ is a threshold graph if and only if it is a split graph and, for split partition $(C, I) $ of $ V $ where $ C $ is a clique and $ I $ is an independent set, there is an ordering $ (x_1, x_2, \ldots, x_p) $ of vertices of $ C $ such that $ N[x_1] \subseteq N[x_2] \subseteq \ldots \subseteq N[x_p],$ and there is an ordering $ (y_1, y_2, \ldots, y_q) $ of the vertices of $ I $ such that $ N(y_1)\supseteq N(y_2) \supseteq \ldots \supseteq N(y_q). $
\begin{thm}
	InSDM is linear time solvable for threshold graphs.
\end{thm}
\begin{proof}
	Let $ G=(V, E) $ be a connected threshold graph with split partition $ (C, I) $. If there exists a vertex $ c_0 \in C $ such that $  N(c_0)\cap I = \emptyset $, then $ S = I \cup \{c_0\}$, otherwise $ S = I $.
	Here, it is clear that $ S $ is an independent set and for every vertex $ u \in V\setminus S $, it can be seen that there exists a vertex $ v \in S $ such that $ (S \setminus \{v\}) \cup \{u\} $ is a dominating set of $ G.$ Hence $ S $ is an InSDS of $ G $ and $ \gamma_{is}(G) \le \vert S \vert. $ \par 
	Let $ D $ be any InSDS of $ G $, then we show that $ \vert D \vert$ $\ge$ $\vert S \vert.$ By contradiction, assume $ \vert D \vert$ $<$ $\vert S \vert.$ Then it follows that there exists a vertex $ v \in D $ for which $ \vert epn(v, D) \vert \ge 2 $ which implies $ D $ is not a SDS of $ G $, a contradiction. Hence $ \vert D \vert$ $\ge$ $\vert S \vert,$ i.e. $ \gamma_{is}\ge \vert S \vert. $ \par
	In a threshold graph, the split partition can be obtained in linear time \cite{threshold1}. Therefore, InSDM in threshold graphs is solvable in linear time.
\end{proof}

\subsection{APX-hardness of independent secure domination}
In this subsection, we prove that the MInSDS problem is APX-hard for graphs with maximum degree $ 5 $. In order to prove this, we use the concept of L-reduction. Let $ IP $ denotes the set of all instances of an optimization problem $ P$, $ SP(x) $ denotes the set of solutions of an instance $ x $ of problem $ P$, $ m_P(x,y)$ denotes the measure of the objective function value for $ x \in IP $ and $ y \in SP(x) $ and $ opt_P(x)$ denotes the optimal value of the objective function $ x \in IP.$ The L-reduction is defined as follows.
\begin{definition}
	Given two NP optimization problems $ F $ and $ G $ and a polynomial time transformation $ f $
	from instances of $ F $ to instances of $ G $, we say that $ f $ is an L-reduction if there are positive constants $ a $ and $ b $ such that for every instance $ x $ of $ F: $
	\begin{enumerate}
		\item $ opt_G(f(x)) \le a * opt_F(x) $.
		\item for every feasible solution $ y $ of $ f(x) $ with objective value $ m_G(f(x), y) = c_2 $ we can in polynomial time find a solution $ y^\prime $ of $ x $ with $ m_F (x, y^\prime) = c_1 $ such that $\vert  opt_F (x) - c_1 \vert \le  b * \vert opt_G(f(x)) - c_2\vert. $
	\end{enumerate}
\end{definition}
\noindent To show APX-hardness of MInSDS, we give an L-reduction from MINIMUM INDEPENDENT DOMINATING SET-$ 3 $ (MInDS-$ 3 $) which has been proved as APX-complete \cite{apxhard}. The MInDS-3 problem is to find minimum InDS of a graph $ G $ with maximum degree $ 3. $
\begin{thm}
	The MInSDS problem is APX-hard for graphs with maximum degree $ 5.$
\end{thm} 
\begin{proof}
	Given an instance $ G=(V,E)$ of maximum degree $ 3 $, where $ V=\{v_1,v_2,\ldots,$ $v_n\} $ of MInDS-3, we construct an instance $ G^\prime=(V^\prime, E^\prime) $ of MInSDS as follows. Let $ P=\{p_1,p_2,\ldots,p_n\} $, $Q=\{q_1,q_2,\ldots,q_n\} $, $R=\{r_1,r_2,\ldots,r_n\} $, $S=\{s_1,s_2,\ldots,s_n\} $ and  $T=\{t_1,t_2,$ $\ldots,t_n\} $. In the graph $ G^\prime,$ $ V^\prime= V$ $\cup$ $P$ $\cup$ $Q$ $\cup$ $R$ $\cup$ $S$ $\cup$ $T$ and $ E^\prime=E \cup \{(v_i,q_i),(q_i,p_i),$ $ (v_i,r_i), (r_i,s_i), (r_i,t_i) : 1 \le i \le n\} .$ Note that $ G^\prime $ is a graph with maximum degree $ 5 $. An example construction of the graph $ G^\prime $ is illustrated in Figure \ref{apx5}.
		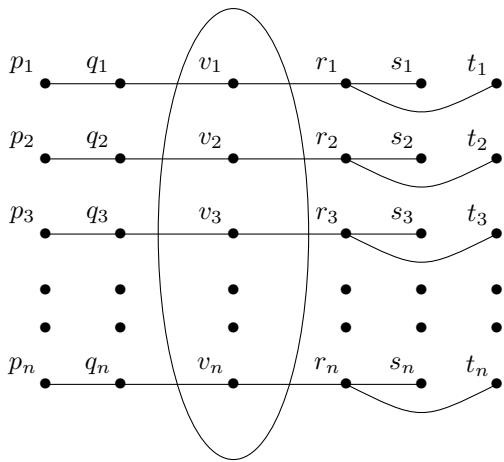
\begin{figure}
			\begin{center}
				\begin{tikzpicture}
				\draw (9,10) ellipse (1cm and 3cm);
%
				
				\node at (9,12){\textbullet}; \node at (8.7,12.25){$v_1$}; 
				\node at (9,11){\textbullet}; \node at (8.7,11.25){$v_2$}; 
				\node at (9,10){\textbullet}; \node at (8.7,10.25){$v_3$}; 
				\node at (9,9.25){\textbullet}; 
				\node at (9,8.75){\textbullet}; 
				\node at (9,8){\textbullet}; \node at (8.7,8.25){$v_n$}; 
				
				\node at (7.5,12){\textbullet}; \node at (7.2,12.25){$q_1$}; 
				\node at (7.5,11){\textbullet}; \node at (7.2,11.25){$q_2$}; 
				\node at (7.5,10){\textbullet}; \node at (7.2,10.25){$q_3$}; 
				\node at (7.5,9.25){\textbullet}; 
				\node at (7.5,8.75){\textbullet}; 
				\node at (7.5,8){\textbullet}; \node at (7.2,8.25){$q_n$};
				
				\node at (6.5,12){\textbullet}; \node at (6.2,12.25){$p_1$}; 
				\node at (6.5,11){\textbullet}; \node at (6.2,11.25){$p_2$}; 
				\node at (6.5,10){\textbullet}; \node at (6.2,10.25){$p_3$}; 
				\node at (6.5,9.25){\textbullet}; 
				\node at (6.5,8.75){\textbullet}; 
				\node at (6.5,8){\textbullet}; \node at (6.2,8.25){$p_n$};  
				
				\node at (10.5,12){\textbullet}; \node at (10.25,12.25){$r_1$};
				\node at (10.5,11){\textbullet}; \node at (10.25,11.25){$r_2$}; 
				\node at (10.5,10){\textbullet}; \node at (10.25,10.25){$r_3$}; 
				\node at (10.5,9.25){\textbullet}; 
				\node at (10.5,8.75){\textbullet}; 
				\node at (10.5,8){\textbullet}; \node at (10.25,8.25){$r_n$};

				\node at (11.5,12){\textbullet}; \node at (11.25,12.25){$s_1$};
				\node at (11.5,11){\textbullet}; \node at (11.25,11.25){$s_2$}; 
				\node at (11.5,10){\textbullet}; \node at (11.25,10.25){$s_3$}; 
				\node at (11.5,9.25){\textbullet}; 
				\node at (11.5,8.75){\textbullet}; 
				\node at (11.5,8){\textbullet}; \node at (11.25,8.25){$s_n$};

				\node at (12.5,12){\textbullet}; \node at (12.25,12.25){$t_1$};
				\node at (12.5,11){\textbullet}; \node at (12.25,11.25){$t_2$}; 
				\node at (12.5,10){\textbullet}; \node at (12.25,10.25){$t_3$}; 
				\node at (12.5,9.25){\textbullet}; 
				\node at (12.5,8.75){\textbullet}; 
				\node at (12.5,8){\textbullet}; \node at (12.25,8.25){$t_n$}; 
				
				\draw (9,12)--(7.5,12)--(6.5,12);
				\draw (9,11)--(7.5,11)--(6.5,11);
				\draw (9,10)--(7.5,10)--(6.5,10);
				\draw (9,8)--(7.5,8)--(6.5,8);								
				
				\draw (9,12)--(11.5,12);
				\draw (9,11)--(11.5,11);
				\draw (9,10)--(11.5,10);
				\draw (9,8)--(11.5,8);
				
				\draw(10.5,12)..controls(11.5,11.5)..(12.5,12); 
				\draw(10.5,11)..controls(11.5,10.5)..(12.5,11); 
				\draw(10.5,10)..controls(11.5,9.5)..(12.5,10);
				\draw(10.5,8)..controls(11.5,7.5)..(12.5,8); 				 				
				\end{tikzpicture}	
				\caption{Example construction of a graph $ G^\prime $ from $ G $} \label{apx5}
			\end{center}
		\end{figure}
First we prove the following claim.
	\begin{claim}
		If $ D^*$ is a minimum InDS of $ G$ and $ S^*$ is a minimum InSDS of $ G^\prime$ then $ \vert S^* \vert = \vert D^* \vert + 3n,$ where $ n=\vert V \vert.$
	\end{claim}
	\begin{proof}
		Suppose $ D^*$ be a minimum InDS of $ G$. Clearly $S^*= D^* \cup S \cup T \cup \{p_i : v_i \in D^* \} \cup \{q_i : v_i \notin D^*\} $ is an InSDS of $ G^\prime.$ Hence $ \vert S^* \vert \le \vert D ^* \vert + 3n.$ \par
		Next we show that $ \vert S^* \vert \ge \vert D ^* \vert + 3n.$ In any InSDS of $ G^\prime$, for each $ i$ where $1 \le i \le n$, we have to choose (i) one vertex from $ \{p_i, q_i\}$ (ii) two vertices from $ \{r_i, s_i, t_i\}.$ Hence if $ D^*$ is a minimum InDS of $ G,$ then any InSDS of $ G^\prime$ must contain at least $ 3n $ new vertices. Therefore, $ \vert S^* \vert \ge \vert D ^* \vert + 3n.$ This completes the proof of claim.
	\end{proof}
	\noindent Let $ D^*$, $ S^* $ be a minimum InDS and a minimum InSDS of $ G $ and $ G^\prime$ respectively. It is known that for any graph $G=(V,E)$ with maximum degree $ \Delta $, $ \gamma(G) \ge \frac{n}{\Delta+1}$, where $ n=\vert V \vert.$ From \cite{Haynes1}, we know that $ \gamma(G) \le i(G)$. Thus, $ \vert D^* \vert \ge \frac{n}{4}.$ From the above claim, it is evident that, $ \vert S^* \vert = \vert D ^* \vert + 3n \le \vert D ^* \vert+12\vert D ^* \vert = 13 \vert D ^* \vert.$ \par
	Now, consider an InSDS $ S^\prime $ of $ G^\prime$. Clearly, the set $ D=S^\prime \cap V$ is a InDS of $ G.$ Since every InSDS of $ G^\prime $ contains at least $ 3n $ vertices which are not part of $ G $, $ \vert D \vert  \le \vert S^\prime \vert - 3n.$ Hence, $ \vert D \vert-\vert D^* \vert  \le \vert S^\prime \vert - 3n-\vert D^*\vert=\vert S^\prime \vert -\vert S^* \vert.$ This implies that there exists an L-reduction with $ a=13 $ and $ b=1.$
\end{proof}
\subsection{Complexity difference in domination and independent secure domination}
Although independent secure domination is one of the several variants of domination problem, these two differ in computational complexity. In particular, there exist graph classes for which the decision version of the first problem is polynomial-time solvable whereas the decision version of the second problem is NP-complete and vice versa.\par
We construct a new class of graphs in which the MINIMUM INDEPENDENT SECURE DOMINATION problem can be solved trivially, whereas the decision version of the DOMINATION problem is NP-complete.
\begin{definition}[GP graph]
	A graph is GP graph if it can be constructed from a connected graph $ G=(V,E)$ where $ \vert V \vert=n$ and $ V =\{v_1,v_2,\ldots,v_n\}, $ in the following way:
	\begin{enumerate}
		\item Create $ n $ copies of path graph each with $ 3 $ vertices i.e. $ P_3$, where $ a_i,b_i$ and $c_i$ are the vertices of $ i^{th} $ copy of path graph $ P_3 $.
		\item Make each $ v_i $ adjacent to vertex $ a_i $ of $ i^{th} $ copy of $ P_3 $.
	\end{enumerate}
\end{definition} 
\begin{center}
	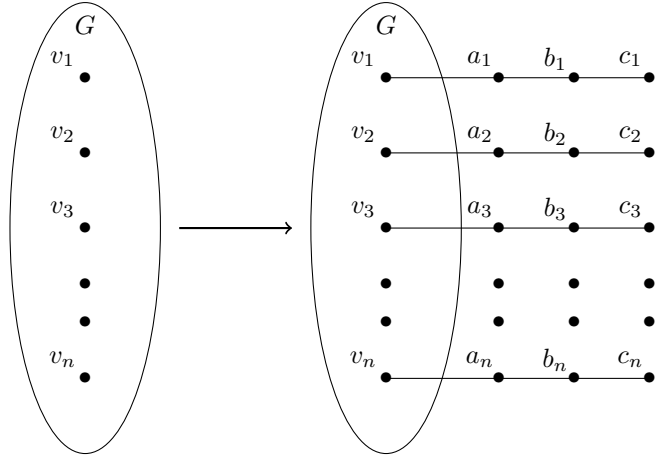
\begin{figure}
		\begin{center}
			\begin{tikzpicture}
			\draw (5,10) ellipse (1cm and 3cm);
			\node at (5,12.7){$ G $};
			\node at (9,12.7){$ G $};
			\draw (9,10) ellipse (1cm and 3cm);
			\node at (5,12){\textbullet}; \node at (4.7,12.25){$v_1$}; 
			\node at (5,11){\textbullet}; \node at (4.7,11.25){$v_2$}; 
			\node at (5,10){\textbullet}; \node at (4.7,10.25){$v_3$}; 
			\node at (5,9.25){\textbullet}; 
			\node at (5,8.75){\textbullet}; 
			\node at (5,8){\textbullet}; \node at (4.7,8.25){$v_n$}; 
			
			\draw [->,thick] (6.25,10) -- (7.75,10);
			
			\node at (9,12){\textbullet}; \node at (8.7,12.25){$v_1$}; 
			\node at (9,11){\textbullet}; \node at (8.7,11.25){$v_2$}; 
			\node at (9,10){\textbullet}; \node at (8.7,10.25){$v_3$}; 
			\node at (9,9.25){\textbullet}; 
			\node at (9,8.75){\textbullet}; 
			\node at (9,8){\textbullet}; \node at (8.7,8.25){$v_n$};

			\node at (10.5,12){\textbullet}; \node at (10.25,12.25){$a_1$};
			\node at (10.5,11){\textbullet}; \node at (10.25,11.25){$a_2$}; 
			\node at (10.5,10){\textbullet}; \node at (10.25,10.25){$a_3$}; 
			\node at (10.5,9.25){\textbullet}; 
			\node at (10.5,8.75){\textbullet}; 
			\node at (10.5,8){\textbullet}; \node at (10.25,8.25){$a_n$};

			\node at (11.5,12){\textbullet}; \node at (11.25,12.25){$b_1$};
			\node at (11.5,11){\textbullet}; \node at (11.25,11.25){$b_2$}; 
			\node at (11.5,10){\textbullet}; \node at (11.25,10.25){$b_3$}; 
			\node at (11.5,9.25){\textbullet}; 
			\node at (11.5,8.75){\textbullet}; 
			\node at (11.5,8){\textbullet}; \node at (11.25,8.25){$b_n$};

			\node at (12.5,12){\textbullet}; \node at (12.25,12.25){$c_1$};
			\node at (12.5,11){\textbullet}; \node at (12.25,11.25){$c_2$}; 
			\node at (12.5,10){\textbullet}; \node at (12.25,10.25){$c_3$}; 
			\node at (12.5,9.25){\textbullet}; 
			\node at (12.5,8.75){\textbullet}; 
			\node at (12.5,8){\textbullet}; \node at (12.25,8.25){$c_n$}; 
			
			\draw (9,12)--(12.5,12);
			\draw (9,11)--(12.5,11);
			\draw (9,10)--(12.5,10);
			\draw (9,8)--(12.5,8);
			
			\end{tikzpicture}	
			\caption{Example GP graph construction} \label{gpgraph}
		\end{center}
	\end{figure}
\end{center}
General GP graph construction is shown in figure \ref{gpgraph}.
\begin{thm}
	If $ G^\prime$ is a GP graph obtained from a graph $ G=(V,E) $ $ (\vert V \vert=n)$, then $ \gamma_{is}(G^\prime)=2n$. 
\end{thm}
\begin{proof}
	Let $ G=(V,E),$ where $V =\{v_1, v_2,\ldots,v_n\} $ be a graph. The construction of  $ G^\prime=(V^\prime, E^\prime) $ as follows. Create $ n $ copies of $ P_3 $, where $ a_i,b_i$ and $ c_i $ are the vertices of $ i^{th} $ copy of $ P_3 $, and create the edges  $ \{(v_i,a_i) : 1 \le i \le n\} $. It is clear that $ G^\prime $ is a GP graph. Let $ S=V \cup  \{b_i : 1\le i\le n\}.$ It can be observed that $ S $ is a InSDS of $ G^\prime$ of size $ 2n$ and hence $ \gamma_{is}(G^\prime) \le 2n.$ \par
	Let $ S $ be any InSDS in $ G^\prime $. Note that $ \vert S \cap \{b_i,c_i : 1\le i \le n\} \vert = n$ and these vertices cannot defend any other vertex in $ V \cup \{a_i : 1\le i \le n\}.$ Therefore, either $ v_i $ or $ a_i $, for each $ i $, where $ 1\le i\le n$ must be included in every InSDS of $ G^\prime$, and hence $ \vert S \vert \ge 2n.$ This completes the proof the theorem.
\end{proof}
\begin{lemma}\label{difflemma}
	Let $ G^\prime$ be a GP graph constructed from a graph $ G=(V,E).$ Then $ G $ has a dominating set of size at most $ k $ if and only if $ G^\prime$ has a dominating set of size at most $ k+n,$ where $ n=\vert V \vert.$
\end{lemma}
\begin{proof}
	Suppose $ D $ be a dominating set of $ G $ of size at most $ k,$ then it is clear that $ D \cup \{b_i : 1 \le i \le n\} $ is a dominating set of $ G^\prime $ of size at most $ n+k.$ \par
	Conversely, suppose $ D^\prime $ is a dominating set of $ G^\prime$ of size $ n+k.$ Then at least one vertex from each pair of vertices $ b_i, c_i $ must be included in $ D^\prime.$ Let $ V^* =\{v_i \in V : a_i \in D^\prime\}$ and $ D^{\prime\prime} = V^* \cup \{(D^\prime \setminus\{a_i\}) : a_i \in  D^\prime \}$. Clearly, $ D^{\prime\prime} \cap V $ is a dominating set of $ G$ of size at most $k$. Hence the lemma.
\end{proof}
\noindent The following result is well known for the DOMINATION DECISION problem.
\begin{thm}(\cite{garey})
	The DOMINATION DECISION problem is NP-complete for general graphs.
\end{thm}
\noindent By using above theorem and Lemma \ref{difflemma} it can be proved that DOMINATION DECISION problem is NP-hard. Hence the following theorem.
\begin{thm}
	The DOMINATION DECISION problem is NP-complete for GP graphs.
\end{thm}
We remark, however, that the two problems, domination and InSDM are not equivalent in computational complexity aspects. A good example is when the input graph is a GP graph, the domination problem is known to be NP-complete whereas the InSDM problem is trivially solvable. Thus, there is a scope to study each of the problems on its own for particular graph classes. Further, it would be interesting to obtain the borderline between tractability and intractability of independent secure domination problem. 


\begin{thebibliography}{99}
	\bibitem{bersto} A.A. Bertossi, {\it Dominating sets for split and bipartite graphs.} Information Processing Letters, \textbf{19} (1984), no. 1, pp.37-40.
	
	
	\bibitem{apxhard} M. Chleb\'ik and J. Chleb\'iko\'v, {\it The complexity of combinatorial optimization problems on d-dimensional boxes.} SIAM Journal on Discrete Mathematics, \textbf{21} (2007), no. 1, pp.158-169.
	
	\bibitem{pog} E.J. Cockayne, P.J.P. Grobler, W.R. Grundlingh, J. Munganga, and J.H. van Vuuren, \emph{Protection of a graph}, Utilitas Mathematica, \textbf{67} (2005), pp. 19-32.
	
	\bibitem{indomb} D.G. Corneil, and Y.Perl, \emph{Clustering and domination in perfect graphs}, Discrete Applied Mathematics, \textbf{9} (1984), pp. 27-39.
	
	\bibitem{courc} B. Courcelle, \emph{The monadic second-order logic of graphs. I. Recognizable sets of finite graphs}, Inform. and Comp. \textbf{85(1)} (1990) 64-75.
	
	\bibitem{dev} A.P. De Villiers, \emph{Edge criticality in secure graph domination}, Ph.D. Dissertation Stellenbosch: Stellenbosch University, (2014).
	
	\bibitem{garey} M.R. Garey, and D.S. Johnson, \emph{Computers and Intractability: A Guide to the Theory of NP-Completeness}, Freeman, New York, (1979).
	
	\bibitem{gol}  M.C. Golumbic, and C.F. Goss, \emph{Perfect elimination and chordal bipartite graphs}, Journal of Graph Theory, \textbf{2} (1978), no. 2, pp. 155-163.
	
	\bibitem{hamid}  I.S. Hamid, and S. Balamurugan, \emph{Isolate domination in graphs}, Arab Journal of Mathematical Sciences, \textbf{22} (2016), no. 2, pp. 232-241.
	
	\bibitem{Haynes1} T.W. Haynes, S.T. Hedetniemi, and P. Slater, \emph{Fundamentals of domination in graphs}, CRC Press, (1998).
	
	\bibitem{Haynes2} T.W. Haynes, S.T. Hedetniemi, and P. Slater, \emph{Domination in graphs: advanced topics}, Marcel Dekker, (1997).
	
	\bibitem{karp}  R.M. Karp, \emph{Reducibility among combinatorial problems}, Complexity of Computer Computations, (1972), pp. 85-103.
	
	\bibitem{threshold1} N.V. Mahadev, and U.N. Peled, \emph{Threshold graphs and related topics}, \textbf{56}, North Holland, (1995).
	
	\bibitem{osd} H.B. Merouane, and M. Chellali, \emph{On secure domination in graphs}, Information Processing Letters, \textbf{1150} (2015), pp. 786-790.
	
	
	\bibitem{rad}  N.J. Rad, \emph{Some notes on the isolate domination in graphs}, AKCE International Journal of Graphs and Combinatorics, \textbf{14} (2017), no. 2, pp. 112-117.
\end{thebibliography}
\end{document}